\documentclass[11pt]{article}
\usepackage{algorithm}
\usepackage{algorithmic}

\usepackage{epsfig,amsthm,amsmath,color}
\usepackage{epsfig,color}

\usepackage{fullpage}
\usepackage{epsfig,amssymb,amsmath,color,framed}

\newtheorem{theorem}{Theorem}[section]

\newtheorem{lemma}[theorem]{Lemma}

\newcommand{\comment}[1]{}
\newcommand{\QED}{\mbox{}\hfill \rule{3pt}{8pt}\vspace{10pt}\par}
{\mbox{}\\[5pt]}

\def\suchthat{\;:\;}
\newcommand{\size}[1]{\left|#1\right|}
\newcommand{\den}[1]{\operatorname{density}\left(#1\right)}
\newcommand{\pr}[1]{{\sf Pr}\left(#1\right)}
\newcommand{\expec}[1]{{\sf E} \left[#1\right]}

\newcommand{\ignore}[1]{}
\newcommand{\eat}[1]{}

\begin{document}
\title{Finding Dense Subgraphs in $G(n,1/2)$}
\author{
Atish Das Sarma\footnote{Work done while at Microsoft Research} \\
Georgia Institute of Technology
\and
Amit Deshpande \\
Microsoft Research
\and
Ravi Kannan \\
Microsoft Research
}
\date{}
\maketitle

\section{Introduction}
Finding the largest clique is a notoriously hard problem, even on random graphs. It is known that the clique number of a random graph $G(n, 1/2)$ is almost surely either $k$ or $k+1$, where $k = \lceil 2\log n - 2\log\log n - 1 \rceil$ (Section 4.5 in \cite{AS}, also \cite{B}). However, a simple greedy algorithm finds a clique of size only $\log n \left(1+o(1)\right)$, with high probability, and finding larger cliques -- that of size even $(1+\epsilon) \log n$ -- in randomized polynomial time has been a long-standing open problem \cite{K}. In this paper, we study the following generalization: given a random graph $G(n, 1/2)$ find the largest subgraph with edge density at least $(1-\delta)$. We show that a simple modification of the greedy algorithm finds a subset of $2 \log n$ vertices whose induced subgraph has edge density at least $0.951$, with high probability. To complement this, we show that almost surely there is no subset of $2.784 \log n$ vertices whose induced subgraph has edge density $0.951$ or more.

We use $G(n, p)$ to denote a random graph on $n$ vertices where each pair of vertices appears as an edge independently with probability $p$. We use $V$ to denote its set of vertices and $E$ to denote its set of edges. Moreover, given two subsets $S \subseteq V$ and $T \subseteq V$, we use $E(S, T)$ to denote the set of edges with one endpoint in $S$ and another endpoint in $T$. The density of the subgraph induced by vertices in $S$ is given by
\[
\den{S} = \frac{\size{E(S, S)}}{{\size{S} \choose 2}}.
\]
Therefore, the expected density of $G(n, 1/2)$ is $1/2$ and the density of any clique is $1$.

In Section~\ref{sec:algo} we describe our algorithm for finding subgraphs of density $1-\delta$. We give a bound on the largest subgraph of density $1-\delta$ in the following Section~\ref{sec:ub}. Finally, in Section~\ref{sec:conclusions}, we present some open problems.

\section{Algorithm for finding large subgraph of density $1-\delta$}
\label{sec:algo}

In this section, we describe our algorithm and give a relationship between the size of the subgraph obtained by the algorithm, and its density. 
In particular, we show that the algorithm can be used to obtain a subset of $2\log n$ vertices of density $0.951$, with high probability. 

\begin{framed}
\noindent{\sc Greedy Algorithm to pick a dense subgraph: } \\
\noindent Input: a random graph $G(n, 1/2)$ and $\delta > 0$. \\
\noindent Output: a subset $S \subseteq V$ of size $k = 2 \log n$.

\begin{enumerate}
\item Partition the vertices into disjoint sets $V = V_{1} \cup V_{2} \cup \dotsb \cup V_{k}$, each of size $n/k$.
\item Initialize $S_{0} = \emptyset$.
\item For $i=0$ to $k-1$ do:
\begin{enumerate}
\item Pick $v_{i+1} \in V_{i+1}$ that has the maximum number of edges to $S_{i}$, i.e.,
\[
v_{i+1} = \underset{{v \in V_{i+1}}}{\operatorname{argmax}} \size{E(v_{i+1}, S_{i})}.
\]
\item $S_{i+1} \leftarrow S_{i} \cup \{v_{i+1}\}$.
\end{enumerate}
\item Return $S = S_{k-1}$.
\end{enumerate}
\end{framed}

Notice that the algorithm first partitions all nodes into $k$ random subsets of the same size, and then picks one vertex from each partition. This partitioning is necessary to argue about independence in our analysis of choosing vertices greedily.

In the analysis below, $H(\delta)$ is the standard notation of the Shannon entropy function, which is $-(\delta\log \delta + (1-\delta)\log (1-\delta))$. The following lemma gives a lower bound on the number of edges we can expect to add to our subgraph, for the $i$-th vertex added by the algorithm.


\begin{lemma} \label{lemma:one-step}
For any $0 \leq i \leq k$ and $\delta_{i}$ that satisfies
\[
H(\delta_{i}) \geq 1 - \frac{1}{i} \log\left(\frac{n}{2k \ln(\log n)}\right),
\]
we have
\[
\pr{\size{E(v_{i+1}, S_{i})} \geq \left(1 - \delta_{i}\right)i} \geq 1 - \frac{1}{\log^{2} n}.
\]
\end{lemma}
\begin{proof}
We know by the previous results, that as long as $k<\log n$, the vertex added has all edges to $S_{k-1}$. Consider $k\geq \log n$. The algorithm has $\frac{n}{l}$ vertices to choose from. The expected number of vertices among these, with at least $(1-\delta_k)k$ vertices is given by,

Fix $v \in V_{i+1}$. The probability that $v$ has at least $(1-\delta_{i})i$ edges to $S_{i}$ is
\[
\pr{\size{E(v, S_{i})} \geq (1-\delta_{i})i} = \sum_{t=(1-\delta_{i})i}^{i} {i \choose t} 2^{-i} = 2^{\left(H(\delta_{i}) + o(1) - 1\right)i},
\]
where $H(\delta) = - \delta \log \delta - (1-\delta) \log (1-\delta)$ is the Shannon entropy (here $\log$ is taken with base $2$). Using independence of these events for different $v \in V_{i+1}$, we get
\begin{align*}
\pr{\size{E(v, S_{i})} < (1-\delta_{i})i,~ \forall v \in V_{i+1}} & \leq \left(1 - 2^{\left(H(\delta_{i}) - 1\right)i}\right)^{n/k} \\
& \leq \left(1 - \frac{2k \ln(\log n)}{n}\right)^{n/k} \\
& \leq \frac{1}{\log^{2} n}.
\end{align*}
Therefore,
\[
\pr{\size{E(v_{i+1}, S_{i})} \geq \left(1 - \delta_{i}\right)i} \geq 1 - \frac{1}{\log^{2} n}.
\]
\end{proof}


We now give a union bound over all $k$ additions of vertices, using the previous lemma.

\begin{lemma} \label{lemma:k-steps}
\[
\pr{\size{E(S, S)} \geq \sum_{i=0}^{k-1} \left(1-\delta_{i}\right)i} \rightarrow 1~ \text{as}~ n \rightarrow \infty.
\]
\end{lemma}
\begin{proof}
Since $V_{1}, V_{2}, \dotsc, V_{k}$ are disjoint, using independence and Lemma \ref{lemma:one-step} we get
\begin{align*}
\pr{\size{E(S, S)} \geq \sum_{i=0}^{k-1} \left(1-\delta_{i}\right)i} & \geq \prod_{i=0}^{k-1} \pr{\size{E(v_{i+1}, S_{i})} \geq \left(1 - \delta_{i}\right)i} \\
& \geq \left(1 - \frac{1}{\log^{2} n}\right)^{k-1} \\
& \geq e^{1/\log n} \qquad \text{using $k = 2\log n$}
\end{align*}
\end{proof}

The point is that we are picking exactly one vertex from each vertex set/partition, and hence do not lose any randomness or independence of the edges. This now gives us a bound on the minimum number of edges one can expect, w.h.p., in the chosen set of $k$ vertices. We are not able to express, in a closed form, the size of a subgraph obtainable using this algorithm for a specific density. Therefore, we state the best density one can guarantee w.h.p. for $k=2\log n$. This is stated as a theorem below, which we prove subsequently.

\begin{theorem}
Our algorithm produces a subset $S \subseteq V$ of size $k = 2 \log n$ such that 
$\den{S} \gtrsim 0.951$, almost surely.
\end{theorem}
\begin{proof}
From Lemma \ref{lemma:k-steps} we have that, almost surely,
\begin{align}
\size{E(S, S)} & \geq \sum_{i=0}^{k-1} \left(1-\delta_{i}\right)i \nonumber \\
& \geq \sum_{i=0}^{k-1} \left(1 - H^{-1}\left(1 - \frac{1}{i} \log\left(\frac{n}{2k \ln(\log n)}\right)\right)\right) i \nonumber \\
& = \sum_{i=0}^{k-1} i - \sum_{i=\log m}^{k-1} i H^{-1}\left(1 - \frac{\log m}{i}\right) \nonumber \\
& = {k \choose 2} - \sum_{i=\log m}^{k-1} i H^{-1}\left(1 - \frac{\log m}{i}\right), \label{eq:edges}
\end{align}
where $m = n/2k \ln(\log n)$. Here we use the fact that we can choose $\delta_{i} = 0$ for the first $\log m$ steps. Now let $k-1 = (1+\alpha) \log m$. Then
\begin{align}
& \sum_{i=\log m}^{k-1} i H^{-1}\left(1 - \frac{\log m}{i}\right) \nonumber \\
& = \sum_{i=\log m}^{(1+\alpha)\log m} i H^{-1}\left(1 - \frac{\log m}{i}\right) \nonumber \\
& = \sum_{t=0}^{\alpha \log m} \left(\log m + t\right) H^{-1}\left(1 - \frac{\log m}{\log m + t}\right) \nonumber \\
& = \log^{2} m \sum_{x=0}^{\alpha} (1+x) H^{-1}\left(1 - \frac{1}{1+x}\right) \nonumber \\
& \leq \log^{2} m \int_{0}^{\alpha} (1+x) H^{-1}\left(1 - \frac{1}{1+x}\right) dx, \label{eq:integral}
\end{align}
Now using Equations \eqref{eq:edges} and \eqref{eq:integral} we have
\begin{align*}
\den{S} & = \frac{\size{E(S, S)}}{{k \choose 2}} \\
& \geq 1 - \frac{\log^{2} m}{{k \choose 2}} \int_{0}^{\alpha} (1+x) H^{-1}\left(1 - \frac{1}{1+x}\right) dx \\
& \geq 1 - \frac{1}{2}\left(1+o(1)\right) \int_{0}^{\alpha} (1+x) H^{-1}\left(1 - \frac{1}{1+x}\right) dx \\
& \gtrsim 0.951.
\end{align*}
using
\[
\alpha = \frac{k}{\log m} - 1 = \frac{2 \log n}{\log n - \log \left(4 \log n \cdot \ln (\log n)\right)} - 1 = 1+o(1).
\]
and computing an upper bound on the integral numerically.
\end{proof}

\section{Upper bound on largest subgraph of density $1-\delta$}
\label{sec:ub}

In this section, we upper bound the size of the largest subgraph of density $1-\delta$ in $G(n,1/2)$.

\begin{theorem}
A random graph $G(n, 1/2)$ has no subgraph of size
\[
\frac{2 \log n + 2 \log e}{1 - H(\delta) - o(1)} + 1
\]
and density at least $1-\delta$, almost surely. In particular, there is no subgraph of size $2.784 \log n$ and density at least $0.951$, almost surely.
\end{theorem}
\begin{proof}
For every $S \subseteq V$ of size $k$, define an indicator random variable $X_{S}$ as follows.
\[
X_{S} = \begin{cases} 1 \quad \text{if $S$ induces a subgraph of density $\geq 1-\delta$} \\ 0 \quad \text{otherwise.} \end{cases}
\]
Thus
\[
\expec{X_{S}} = \sum_{i=(1-\delta){k \choose 2}}^{{k \choose 2}} {{k \choose 2} \choose i} 2^{- {k \choose 2}} = 2^{\left(H(\delta) + o(1) - 1\right) {k \choose 2}}.
\]
By linearity of expectation, the expected number of subgraphs of size $k$ and density at least $1-\delta$ is
\begin{align*}
\expec{\sum_{S \suchthat \size{S}=k} X_{S}} & = \sum_{S \suchthat \size{S}=k} \expec{X_{S}} \\
& = {n \choose k} 2^{\left(H(\delta) + o(1) - 1\right) {k \choose 2}} \\
& \leq \left(\frac{en}{k}\right)^{k} \left(2^{\left(H(\delta) + o(1) - 1\right) \frac{k-1}{2}}\right)^{k} \\
& = \left(\frac{en}{k} \cdot 2^{\left(H(\delta) + o(1) - 1\right) \frac{k-1}{2}}\right)^{k} \\
& = \left(\frac{2^{\left(1 - H(\delta) - o(1)\right) \frac{k}{2}}}{k} \cdot 2^{\left(H(\delta) + o(1) - 1\right) \frac{k-1}{2}}\right)^{k} \\
& = \left(\frac{2^{\left(1 - H(\delta) + o(1)\right)/2}}{k}\right)^{k} \rightarrow 0, \quad \text{as $n \rightarrow \infty$},
\end{align*}
using
\[
k = \frac{2 \log n + 2 \log e}{1 - H(\delta) - o(1)} + 1.
\]
Therefore, by Markov inequality we have
\[
\pr{\sum_{S \suchthat \size{S}=k} X_{S} \geq 1} \leq \expec{\sum_{S \suchthat \size{S}=k} X_{S}} \rightarrow 0,
\]
as $n \rightarrow \infty$. Or in other words, almost surely there is no subset of $k$ vertices that induce a subgraph of density at least $1-\delta$.
\end{proof}

Notice that for density 0.951, the gap/ratio between the largest subgraph that exists and the largest subgraph that we can find is smaller than in the case of cliques. This is interesting, although not entirely unexpected as for density $0.5$, the whole graph can be output. This ratio for density 0.951 is however significantly smaller than 2; it is 2.784/2 = 1.392. 

\section{Conclusions}
\label{sec:conclusions}

For a concrete open problem, is there a polynomial time algorithm that outputs a subgraph of density $1-\epsilon$ and size $2\log n$ for any choice of $\epsilon > 0$ ?

Are there simple algorithms that beat the density bound of $0.95$ for subgraphs of size $2\log n$. Is there an $O(n^{\log n})$ time algorithm that finds the largest clique in $G(n,1/2)$? If not, what is the maximum density obtainable for a subgraph of size $2\log n$? Spectral techniques could be tried.

%

\bibliographystyle{amsplain}
\bibliography{clique-ref}

\providecommand{\bysame}{\leavevmode\hbox to3em{\hrulefill}\thinspace}
\providecommand{\MR}{\relax\ifhmode\unskip\space\fi MR }
\providecommand{\MRhref}[2]{%
  \href{http://www.ams.org/mathscinet-getitem?mr=#1}{#2}
}
\providecommand{\href}[2]{#2}
\begin{thebibliography}{1}

\bibitem{AS}
Noga Alon and Joel Spencer, \emph{The probabilistic method (2nd edition)},
  second ed., Wiley Interscience, 2000.

\bibitem{B}
B\'{e}la Bollob\'{a}s, \emph{Random graphs}, Academic Press, New York, 1985.

\bibitem{K}
Richard Karp, \emph{The probabilistic analysis of some combinatorial search
  algorithms},  (1976), 1--19.

\end{thebibliography}
\end{document}